\documentclass[10pt,onecolumn]{IEEEtran}  % Comment this line out
\IEEEoverridecommandlockouts                              % This command is only
\usepackage{fancyhdr}
\usepackage{amssymb}
\usepackage{amsmath}
\usepackage{amsthm}
\usepackage{graphicx}
\usepackage{listings}
\usepackage{float}
\usepackage{enumerate}
\usepackage{epstopdf}
\usepackage{color}
\usepackage{setspace}
\usepackage{bm}
\usepackage{cite}
\usepackage{setspace}
\usepackage{cite}
\usepackage{dsfont}
\usepackage{enumitem}
\usepackage{tabu}
\usepackage[mathscr]{euscript}
\usepackage{threeparttable}
\usepackage{multirow}
\include{hollyMacros}
\newcommand{\sbal}{{S_{\rm b}}}

\newcommand{\sbar}{\bar{S}}
\newcommand{\abar}{\bar{a}}
\newcommand{\ax}{\alpha^{\rm resil}}
\newcommand{\ay}{\alpha^{\rm sus}}
\newcommand{\accord}{\in}
\newcommand{\sss}{\mathcal{S}}

\newcommand{\FID}{{\rm FI}}

\newcommand{\UR}{{\rm UR}}
\newcommand{\FI}{{\rm FI}}
\newcommand{\MI}{{\rm MI}}
\newcommand{\type}{{\rm {\bf type}}}
\newcommand{\aaa}{\mathcal{A}}

\newcommand{\geer}{\mathcal{G}^{\rm r}}

\newcommand{\eps}{\epsilon}
\newcommand{\N}{\mathbb{N}}
\newcommand{\R}{\mathbb{R}}
\newcommand{\st}{:}
\newtheorem{theorem}{Theorem}
\newtheorem{defn}{Definition}
\newtheorem{lemma}{Lemma}[theorem]

\newtheorem{corollary}[theorem]{Corollary}

\normalsize\title{\LARGE \bf
Security Against Impersonation Attacks in Distributed Systems
\thanks{This research was supported by ONR grant \#N00014-17-1-2060 and NSF grant \#ECCS-1638214.}}

\author{
Philip N. Brown,
Holly P. Borowski, and
Jason R. Marden\thanks{P. N. Brown (corresponding author) and J. R. Marden are with the Department of Electrical and Computer Engineering, University of California, Santa Barbara, \texttt{\{pnbrown,jrmarden\}@ece.ucsb.edu}.}
\thanks{H. P. Borowski is a research scientist at Numerica Corporation, \texttt{hollyboro@gmail.com}}
\thanks{A preliminary version of this paper appeared in~\cite{Borowski2015}.}}

\begin{document}
\graphicspath{{figures/}}

\maketitle

\begin{abstract}
In a multi-agent system, transitioning from a centralized to a distributed decision-making strategy can introduce vulnerability to adversarial manipulation.
We study the potential for adversarial manipulation in a class of graphical coordination games where the adversary can pose as a friendly agent in the game, thereby influencing the decision-making rules of a subset of agents.
The adversary's influence can cascade throughout the system, indirectly influencing other agents' behavior and significantly impacting the emergent collective behavior.
The main results in this paper focus on characterizing conditions under which the adversary's local influence can dramatically impact the emergent global behavior, e.g., destabilize efficient Nash equilibria.

\end{abstract}

\section{Introduction}

Engineering and social systems often consist of many agents making decisions based on locally available information.
For example, a group of unmanned aircraft performing surveillance in a hostile area may use a distributed control strategy to limit communication and thus remain undetected.
Social systems are inherently distributed: individuals typically make decisions based on personal objectives and in response to the behavior of friends and acquaintances.
For example, the decision to adopt a recently released technology, such as a new smartphone, may depend both on the quality of the item itself and on friends' choices~\cite{Kreindler2014}.

In a distributed multiagent system, a system designer endows each agent with a set of decision-making rules; once the system is in operation, the designer is no longer available to modify agent algorithms in response to changing conditions.
Thus, an adversary that arrives on the scene while the system is operating may be able to exploit the agents' decision-making algorithms to influence the overall behavior of the system.

Much work in the area of security in cyber-physical systems has focused on reducing the potential impact of adversarial interventions through detection mechanisms: detection of attacks in power networks, estimation and control with corrupt sensor data, and monitoring \cite{Hendrickx2015,Fawzi2014, Bai2014,Pasqualetti2012,Hota2017}.
In contrast to this research, our work focuses on characterizing the impact an adversary may have on distributed system dynamics when no mitigation measures are in place.

We use graphical coordination games, introduced in \cite{Cooper1999, Ullmann-Margalit1977}, to study the impact of adversarial manipulation.
The foundation of a graphical coordination game is a simple two agent coordination game, where each agent must choose one of two alternatives, $\{x,y\}$, with payoffs depicted by the following payoff matrix which we denote by $u(\cdot)$: 

\vspace{-2mm}
\begin{equation}
\setlength{\extrarowheight}{2pt}
\mbox{
\begin{tabular}{r|c|c|}
\multicolumn{1}{r}{}	&\multicolumn{1}{c}{$x$}	&\multicolumn{1}{c}{$y$}\\
\cline{2-3}$x$			&$1+\alpha,\,1+\alpha$	&$0,\,0$\\
\cline{2-3}$y$			&$0,\,0$				&$1,\,1$\\
\cline{2-3}
\end{tabular}}\label{eq:matrix}
\end{equation}
%\vspace{2mm}
%
Both agents prefer to agree on a convention, i.e., $(x,x)$ or $(y,y)$, than disagree, i.e., $(x,y)$ or $(y,x)$, with a preference for agreeing on $(x,x)$.
The parameter $\alpha > 0$ indicates that $(x,x)$ has an intrinsic advantage has over $(y,y)$; we refer to $\alpha$ as the \emph{payoff gain}.
Nonetheless, unilaterally deviating from $(y,y)$ for an individual agent incurs an immediate payoff loss of $1$ to $0$; hence, myopic agents may be reluctant to deviate, stabilizing the inefficient equilibrium $(y,y)$.

This two player coordination game can be extended to an $n$-player \emph{graphical coordination game}\cite{Kearns2001,Young2011, Montanari2010}, where the interactions between the agents $N=\{1, 2, \dots, n\}$ are described by an undirected graph $G = (N,E)$, where an edge $(i,j)\in E$ for some $i\neq j$ indicates that agent $i$ is playing the two-player coordination game~\eqref{eq:matrix} with agent $j$.
An agent's total payoff is the sum of payoffs it receives in the two-player games played with its neighbors ${\cal N}_i = \{j \in N : (i,j) \in E\}$, i.e., for a joint action $a = (a_1, \dots, a_n) \in \{x,y\}^n$, the utility function of agent $i$ is
\begin{equation}\label{e:original utility}
U_i(a_1, \dots, a_n) = \sum_{j \in {\cal N}_i} u(a_i, a_j),
\end{equation}
where $u(\cdot)$ is chosen according to payoff matrix~\eqref{eq:matrix}.
Joint actions $\vec{x}:=(x,x,\ldots,x)$ and $\vec{y}:=(y,y,\ldots,y)$, where either all players choose $x$ or all players choose $y$, are Nash equilibria of the game for any graph;
other equilibria may also exist depending on the structure of graph $G.$

The system operator's goal is to endow agents with decision-making rules to ensure that for any realized graph $G$ and payoff gain $\alpha$, the emergent behavior maximizes the sum of agents' utilities.
Log-linear learning \cite{Blume1993, Shah2010} is one distributed decision making rule that selects the efficient equilibrium in this setting.%
\footnote{
It is worth mentioning that even in more complicated graphical coordination games, log-linear learning selects the optimal action profile as stochastically stable.
For example, in a graphical coordination game if some agents' actions are fixed at $y$, some action profile other than $\vec{x}$ might be optimal, and log-linear learning would select it.
Thus, our paper considers an important special case of a more general model, and we expect that the lessons learned here will inform work on more general settings.
} 
Although agents predominantly maximize their utilities under log-linear learning, selection of the efficient equilibrium is achieved by allowing agents to choose suboptimally with some small probability, specified in Section~\ref{sec:model} in~\eqref{e:LLL dynamics}.

Log-linear learning has been studied for prescribing control laws in many distributed engineering systems \cite{Marden2013,Zhu2009,Goto2010,Staudigl2012,Fox2010}, as well as for analyzing the emergence of conventions in social systems \cite{Young1993,Shah2010}.
The equilibrium selection properties of log-linear learning extend beyond coordination games to the class of potential games \cite{Monderer1996}, which often can be used to model engineering systems where the efficient Nash equilibrium is aligned with the optimal system behavior \cite{Marden2007,Marden2013,Wolpert2001}.
This prompts the question: can adversarial manipulation  alter the emergent behavior of log-linear learning in the context of graphical coordination games (or more broadly in distributed engineering systems)?

We assume that an adversary wishes to influence agents to play the less efficient Nash equilibrium~$\vec{y}$.
In this context, the adversary could model a political rival, attempting to sway voters to adopt an alternative political viewpoint; or, the adversary could model an advertiser that wishes to sway consumers to purchase a competing product.
We model the adversary as additional nodes/edges in the graph, where the new node plays a fixed action $y$ in an effort to influence agents' utilities and provide them an additional incentive to play $y$.

We investigate the tradeoffs between the amount of information available to an adversary, the policies at the adversary's disposal, and the adversary's resulting ability to stabilize the alternative Nash equilibrium~$\vec{y}$.
For concreteness, we perform this analysis by specifying three distinct styles of adversarial behavior:
\begin{itemize}
\item \emph{Uniformly Random:} This type of adversary influences a random subset of agents at each time step.
Uniformly random adversaries have the least amount of information available, essentially only requiring knowledge of $n$. 
\item \emph{Fixed Intelligent}: This type of adversary chooses a subset of agents to influence; this subset is fixed for all time. 
Fixed intelligent adversaries know the entire graph structure~$G$.
\item \emph{Mobile Intelligent}: This type of adversary can choose which agents to influence as a function of the current joint action profile.
Thus, mobile intelligent adversaries know the graph structure, and at each time step, the action choices of all agents. 
\end{itemize}

Our results include an initial study on the influence of uniformly random and fixed intelligent agents on general graphs, as well as a complete characterization of each adversary's ability to stabilize~$\vec{y}$ in a ring graph.

\section{Model}\label{sec:model}

\subsection{Model of agent behavior}

Suppose agents in $N$ interact according to the graphical coordination game above, specified by the tuple $(G,\alpha)$, with underlying graph $G = (N,E)$, alternatives $\{x,y\}$, and payoff gain $\alpha\in\mathbb{R}$.
We denote the joint action space by $\mathcal{A} = \{x,y\}^n,$ and we write
$$(a_i,a_{-i}) = (a_1,a_2,\ldots,a_i,\ldots,a_n)\in\mathcal{A}$$
when considering agent $i$'s action separately from other agents' actions.

A special type of graph considered in some of our results is a \emph{ring} graph, defined as follows.
Let ${G} = (N,E)$ with $N = \{1,2,\ldots,n\}$ and $E = \left\{\{i,j\}:j = i+1 \mbox{ mod } n\right\},$ i.e., ${G}$ is a ring (or cycle) with $n$ nodes.%
\footnote{
When considering ring graphs, all addition and subtraction on node indices is assumed to be ${\rm mod }\ n$.
}
We denote the set of all ring graphs by $\geer$.

Now, suppose agents in $N$ update their actions according to the \emph{log-linear learning} algorithm at times $t = 0,1,\ldots,$ producing a sequence of joint actions $a(0),a(1),\ldots.$
We assume agents begin with joint action, $a(0)\in\mathcal{A}$, and let $a(t) = (a_i,a_{-i})\in\mathcal{A}.$
At time $t\in\mathbb{N}$, an agent $i\in N$ is selected uniformly at random to update its action for time $t+1$; all other agents' actions will remain fixed.
Agent $i$ chooses its next action probabilistically according to:%
\footnote{Agent $i$'s update probability is also conditioned on the fact that agent $i$ was selected to revise its action, which occurs with probability $1\mathop{/}n$.
For notational brevity we omit this throughout, and  $\Pr[a_i(t+1) = A\,|\, a_{-i}(t) = a_{-i}]$, for example, is understood to mean $\Pr[a_i(t+1) = x\,|\, a_{-i}(t) = a_{-i},\, i \text{ selected for update}].$
}
\begin{align}
\Pr[a_i(t+1) = x\,|\, a_{-i}(t) = a_{-i}] %\nonumber \\
= {\exp \left(\beta \cdot U_i(x,a_{-i})\right)\over \exp \left(\beta \cdot  U_i(x,a_{-i}\right)+\exp \left(\beta \cdot U_i(y,a_{-i})\right)}.\label{e:LLL dynamics}
\end{align} 
Parameter $\beta>0$ dictates an updating agent's degree of rationality and is identical for all agents $i\in N$. 
As $\beta\to \infty$, agent $i$ is increasingly likely to select a utility-maximizing action, and as $\beta\to 0$, agent $i$ tends to choose its next action uniformly at random.
The joint action at time $t+1$ is $a(t+1) = (a_i(t+1),a_{-i}(t)).$

Joint action $a\in\mathcal{A}$ is \emph{strictly stochastically stable} \cite{Foster1990} under log-linear learning dynamics if, for any $\eps >0$, there exist $B<\infty$ and $T<\infty$ such that
\begin{equation} \label{eq:weirdSSS}
\Pr[a(t)=a] > 1-\eps,\quad\text{for all } \beta > B, t>T
\end{equation}
where $a(t)$ is the joint action at time $t\in\N$ under log-linear learning dynamics.

Joint action $\vec{x}$ is strictly stochastically stable under log-linear learning for any graphical coordination game whenever $\alpha>0$\cite{Blume1993}.
We will investigate conditions when an adversary can destabilize $\vec{x}$ and stabilize the alternative coordinated equilibrium~$\vec{y}$.

\subsection{Model of adversarial influence}
Consider the situation where agents in $N$ interact according to the graph $G$ and update their actions according to log-linear learning, and an adversary seeks to  convert as many agents in $N$ to play action $y$ as possible.
At each time $t\in \N$ the adversary influences a set of agents $S(t)\subseteq N$ by posing as a friendly agent who always plays action $y$.
Agents' utilities, $\tilde{U}:\mathcal{A}\times 2^N\to\R$, are now a function of adversarial and friendly behavior, defined by:
\begin{equation}\label{e:new utility}
\tilde{U_i}((a_i,a_{-i}),S) = %U_i(a_i,a_{-i}) + \mathds{1}_{i\in S, a_i = y}
\begin{cases}
U_i(a_i,a_{-i})	&\text{if } i\notin S\\
U_i(a_i,a_{-i})	&\text{if } a_i = x\\
U_i(a_i,a_{-i}) + 1 &\text{if } i\in S,  a_i = y 
\end{cases}
\end{equation}
where $(a_i,a_{-i})\in\mathcal{A}$ represents friendly agents' joint action, and \emph{influence set} $S\subseteq N$ represents the set of agents influenced by the adversary.
If $i\in S(t)$, agent $i$ receives an additional payoff of 1 for coordinating with the adversary at action $y$ at time $t\in\N$; that is, to agents in $S(t)$, the adversary appears to be a neighbor playing action $y$.
By posing as a player in the game, the adversary can manipulate the utilities of agents belonging to $S$, providing an extra incentive to choose the inferior alternative~$y$.

Throughout, we write $k$ to denote the number of friendly agents the adversary can connect to, called the adversary's \emph{capability}.
Given $k$, $\sss_k:=\left\{ S\in2^N, |S|=k \right\}$ denotes the set of all possible influence sets.
In this paper, we consider three distinct models of adversarial behavior, which we term fixed intelligent (FI), mobile intelligent (MI), and uniformly random (UR).
To denote a situation in which influence sets $S(t)$ are chosen by an adversary of $\type\in\{\rm \FI,\MI,\UR\}$ for a given $k$, we write $S\accord\type(k)$.

If an adversary is fixed intelligent (FI), this means that the influence set $S$ is a function only of the graph structure and $\alpha$.
That is, the adversary must commit to an influence set $S$ that is fixed for all time (in the following, note that $S$ is always implicitly assumed to be a function of $G$):
\begin{equation}
S\accord\FI(k) \implies  S(t) = S. \label{eq:fi}
\end{equation}
If an adversary is mobile intelligent (MI), this means that the influence set $S(a)$ is a function of the graph structure, $\alpha$, \emph{and} $a(t)$, the state at time $t$:
\begin{equation}
S\accord\MI(k) \implies  S(t) = S(a(t)). \label{eq:mi}
\end{equation}
Note that type-MI adversaries have the freedom to choose a mapping $S:\aaa\rightarrow\sss_k$, whereas a type-FI adversary must choose that mapping to be a constant function of state $a$.
Finally, if the adversary is uniformly random (UR), the influence set $S$ at each time $t$ is chosen uniformly at random from $\sss_k$, independently across time:
\begin{equation}
S\accord\UR(k) \implies  S(t) \sim {\rm unif}\{\sss_k\}. \label{eq:ur}
\end{equation}

\subsection{Resilience metrics}

Given nominal game $(G,\alpha)$, adversary policy $S\accord\type(k)$ for some $\type$ and $k\geq1$, we write the set of stochastically stable states associated with log-linear learning as
\begin{equation}
{\rm SS}\left(G,\alpha,S\right).
\end{equation}

To quantify how the payoff gain $\alpha$ affects various adversary types' ability to influence behavior, we say that a game $(G,\alpha)$ is \emph{resilient} to adversarial influence of $\type(k)$ if for every $S\in\type(k)$, $SS(G,\alpha,S)=\vec{x}$.
A quantity measuring the resilience of a particular graph structure $G$ is then
\begin{equation}
\ax(G,\type(k)) \triangleq  \inf \left\{\alpha : (G,\alpha) \mbox{ is resilient to }\type(k)\right\},
\end{equation}
so that whenever $\alpha>\ax(G,\type(k))$, then $\vec{x}$ is the strictly stochastically-stable state regardless of what a $\type(k)$-adversary does.

Similarly, we say a game $(G,\alpha)$ is \emph{susceptible} to adversarial influence of $\type(k)$ if there exists a policy $S\in\type(k)$ for which $SS(G,\alpha,S)=\vec{y}$.
A quantity measuring the susceptibility of a particular graph structure $G$ is then
\begin{equation}
\ay(G,\type(k)) \triangleq  \sup \left\{\alpha : (G,\alpha) \mbox{ is suscep. to }\type(k)\right\},
\end{equation}
so that whenever $\alpha<\ay(G,\type(k))$, then by employing the right policy, a $\type(k)$-adversary can ensure that $\vec{y}$ is strictly stochastically stable.
Note that for a fixed graph $G$ and $\type(k$), it is always true that $\ay(G,\type(k))\leq\ax(G,\type(k))$.

\subsection{Summary of results}

In the following sections, we provide an initial study into the effects of adversarial influence on emergent behavior in graphical coordination games.
We give conditions under which a fixed intelligent adversary can stabilize $\vec{y}$ (or be prevented from doing so) in a general graph, as well as a general theorem characterizing the impact of $k$ on a general graph's susceptibility to a uniformly random adversary.

Since the analysis for general graphs rapidly becomes intractable, we then derive exact susceptibilities to each type of adversary for ring graphs.
These susceptibilities for ring graphs are summarized below.
\begin{align*}
\ay\left(G;\FI(k)\right) 	&= \frac{k}{n}, 		&k\leq n 		&&\mbox{(Theorem~\ref{thm:ringFI})} \\
\ay\left(G;\UR(k)\right) 	&= \frac{1}{2},		&1\leq k\leq n-1	&&\mbox{(Theorem~\ref{thm:mr-ring})}\\
\ay\left(G;\MI(k)\right) 	&= \frac{k}{k+1},	&k\in\{1,2\}	&&\mbox{(Theorem~\ref{thm:MIring})}\\
\ay\left(G;\MI(k)\right) 	&= \frac{n-1}{n},	&3\leq k\leq n-1&&\mbox{(Theorem~\ref{thm:MIring})}
\end{align*}

These results are also depicted in Figure~\ref{fig:comparison} for $n=10$.
Naturally, for a given $k$, mobile intelligent adversaries are always at least as effective as fixed intelligent adversaries.
However, the intuition behind susceptibility to uniformly-random adversaries is less obvious, as Figure~\ref{fig:comparison} suggests that $k$ has no impact on the effectiveness of this type of adversary.
Indeed, our Theorem~\ref{thm:ur-gen} shows that this is a generic fact: the susceptibility of any graph to uniformly-random adversaries is constant for all $k\in\{1,\ldots,n-1\}$.

\section{Fixed intelligent adversarial influence} \label{sec:fi}

In the fixed intelligent model of adversarial behavior with capability $k$, the adversary chooses a fixed subset $S\in\sss_k$ of agents to influence, as in~\eqref{eq:fi}.
In a sense, this is the type of adversary that is most limited, as the adversary has no ability to react to changing conditions as the agents update their actions.
Nonetheless, as can be seen from Figure~\ref{fig:comparison}, fixed intelligent adversaries actually can outperform uniformly-random adversaries if $k$ is sufficiently large.

\begin{figure}
% this figure generated by https://github.ucsb.edu/philipbrown/adversary-python/blob/master/plotscript.py
  \centering
    \includegraphics[scale=.25]{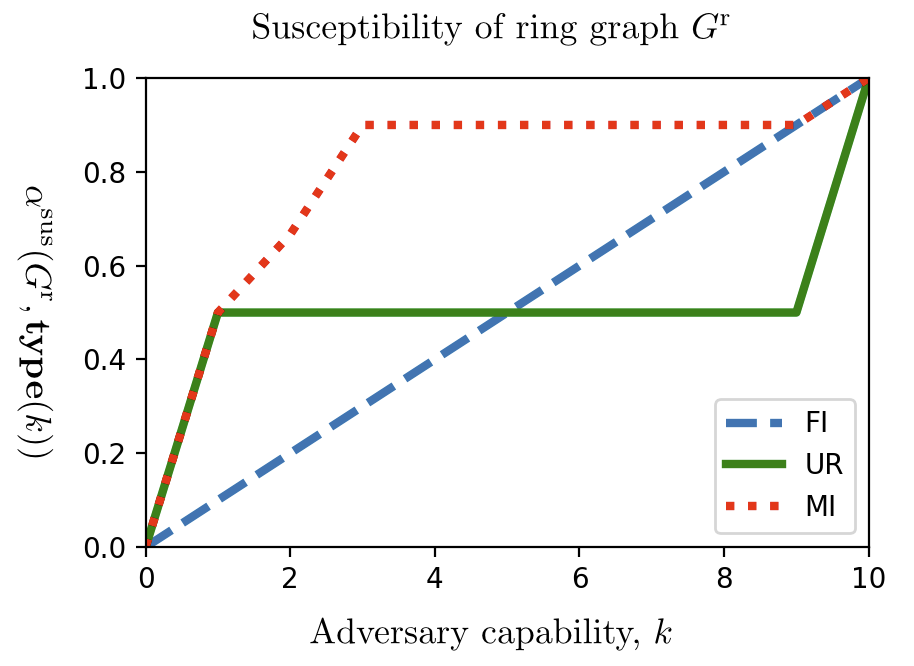}
  \caption{Values of $\alpha$ below which each type of adversary can stabilize joint action $\vec{y}$ in a $10$-agent ring graph as a function of adversary capability $k$.
  For uniformly random (UR) adversaries, this shows that $k=1$ yields the same susceptibility as $k=n-1$, but that for the intelligent types of adversaries, susceptibility depends strongly on $k$.}
  \label{fig:comparison}
\end{figure}

\subsection{General graphs}

We begin with a theorem which provides tools to compute an exact characterization of a fixed intelligent adversary's influence on a general graph.
Theorem~\ref{thm:genFI} gives necessary and sufficient conditions that a graph is either resilient against or susceptible to a given influence set $S$.
We require the following definition:
\begin{equation}
d(S,T) := |\{\{i,j\}\in E\st i\in S, j\in T\}|.
\end{equation}

\begin{theorem}\label{thm:genFI}
For any graph $G$, let a fixed intelligent adversary influence set $S\in\sss_k$.
\begin{enumerate}
\item Joint action $\vec{x}$ is strictly stochastically stable if and only if $\forall T\subseteq N,\,T\neq\emptyset$,
\begin{equation}
\alpha > {|T\cap S| - d(T,N\setminus T)\over d(T,N)}. \label{eq:thm1}
\end{equation}

\item Joint action $\vec{y}$ is strictly stochastically stable if and only if $\forall T\subset N,\,T\neq N$,
\begin{equation}\label{eq:thm2}
\alpha < {d(T,N\setminus T) + k - |T\cap S|\over d(N\setminus T,N\setminus T)}.
\end{equation}
\end{enumerate}
\end{theorem}
\noindent The proof of Theorem~\ref{thm:genFI} appears in the Appendix.

In principle, Theorem~\ref{thm:genFI} provides for the computation of the resilience and susceptibility of any graph.
However, the computational complexity of doing so for an arbitrary graph is at least exponential in $n$, as the stated conditions have to be checked for all $T\subset N$ and then again potentially for all $S$.
Nonetheless, it is possible to compute bounds on resilience and susceptibility from these expressions in some special cases.

For example,~\eqref{eq:thm2} allows us to state a general upper bound on the susceptibility of any graph for any $k$, given in the following corollary:

\begin{corollary} \label{cor:fid-general}
Let $G=(V,E)$ be any graph, and let $k\leq n-1$:
\begin{equation}
\ay(G,\FID(k)) \leq \frac{k}{|E|}. \label{eq:fid-bound}
\end{equation}
\end{corollary}
\begin{proof}
The result is obtained from~\eqref{eq:thm2} by setting $T=\emptyset$.
In this case, $N\setminus T=N$, and $d(N,N)=|E|$ by definition.
\end{proof}

An insight we can glean from Corollary~\ref{cor:fid-general} is that the susceptibility of a graph is related in some sense to its sparsity: graphs with many edges may tend to be less susceptible to fixed adversarial influence than graphs with few edges.
Note that~\eqref{eq:fid-bound} is not known to be tight for general graphs, but we show in the next theorem that it is tight for ring graphs for all~$k$.

\subsection{Exact susceptibilities for ring graphs}

We now analyze a fixed intelligent adversary's influence on a ring graph. 
Define 
\begin{align*}
%[t] := \{1,2,\ldots,t\}\subseteq N, \text{ and }
[i,j] := \{i,i+1,\ldots,j\}\subseteq N,
\end{align*}
and recall that in ring graphs, all addition and subtraction on node indices is assumed to be ${\rm mod}\ n$.

Theorem~\ref{thm:ringFI} gives the susceptibility of any ring graph influenced by a fixed intelligent adversary.

\begin{theorem}\label{thm:ringFI}
Let $G^{\rm r}\in\geer$ be a ring graph that is influenced by a fixed intelligent adversary with capability $k\leq n$.
Then
\begin{equation}
\ay\left(G^{\rm r},\FI(k)\right) = \frac{k}{n}. \label{eq:FIsus}
\end{equation}
This can be realized by an adversary distributing its influence set $S$ as evenly as possible around the ring, so that 
\begin{equation}
\left|S\cap [i,i+t]\right| \leq \left\lceil {kt\over n}\right\rceil \label{eq:FIpolicy}
\end{equation}
for any set of nodes $[i,i+t]\subseteq N$, with $i\in N$, $t\leq n$.

\end{theorem}
\noindent The proof of Theorem~\ref{thm:ringFI} is in Appendix~\ref{a:fixed line graph proofs}.

\section{Uniformly random adversarial influence}

\subsection{General graphs}

It may be difficult to characterize the susceptibility of an arbitrary graph to random adversarial influence, but the following theorem gives an important piece of the puzzle.
Here, we show perhaps counterintuitively that that the susceptibility of every graph to a uniformly random adversary is independent of the adversary's capability $k$.
\begin{theorem}\label{thm:ur-gen}
Let $G$ be any graph.
For any $k\in\{2,\ldots,n-1\}$,
\begin{equation}
\ay\left(G,\UR(k)\right) = \ay\left(G,\UR(1)\right).
\end{equation}
\end{theorem}
The proof of Theorem~\ref{thm:ur-gen} appears in the Appendix.
This result means that each graph has a universal threshold such that if $\alpha$ falls below this, then even a single uniformly-random adversary will eventually influence all agents to play $\vec{y}$.
Note that larger $k$ likely allows the adversary to achieve $\vec{y}$ \emph{more quickly,} but that the threshold value itself for $\alpha$ is independent of~$k$.
We study this distinction numerically in Section~\ref{sec:sims}.

\subsection{Exact susceptibility for ring graphs}

When more graph structure is known, it may be possible to derive precise expressions for $\ay$.
In this section, we consider an adversary which influences a ring graph uniformly at random according to~\eqref{eq:ur}.

\begin{theorem}\label{thm:mr-ring}
Let $G^{\rm r}\in\geer$ be a ring graph that is influenced by a uniformly random adversary with capability $k\leq n$.
For all $k$, it holds that $\ax(G^{\rm r},\UR(k))=\ay(G^{\rm r},\UR(k))$.
Furthermore,
\begin{equation}
\ay\left(G^{\rm r},\UR(k)\right) = \left\{
\begin{array}{cl}
%0 			& \mbox{if } k=0, \vspace{2mm} \\ 
\frac{1}{2} 	& \mbox{if } k\in\{1,\ldots,n-1\},\vspace{2mm} \\
1 			& \mbox{if } k=n.
\end{array}\right.
\end{equation}

\end{theorem}
\noindent Theorem~\ref{thm:mr-ring} is proved in Appendix~\ref{a:mobile random}.

Consider Theorems~\ref{thm:ringFI} and~\ref{thm:mr-ring} from the point of view of an adversary, and suppose that an adversary cannot choose $k$, but can choose whether to be fixed intelligent or uniformly random.
Theorems~\ref{thm:ringFI} and~\ref{thm:mr-ring} suggest that if the adversary's capability is low, it is better to employ a uniformly-random strategy than a fixed one; on the other hand, the conclusion is reversed if the capability is high.
Note that stochastic stability is a measure of behavior in the very long run, and as such it may not be a good indicator of short-term behavior.
In Section~\ref{sec:sims}, we present simulations which suggest that in the short-term, in some situations, uniformly random adversaries outperform fixed intelligent ones regardless of susceptibility.

\section{Mobile intelligent adversarial influence on ring graphs}

Finally we consider type $\MI(k)$ adversaries on ring graphs.
Recall that mobile intelligent adversaries choose influence set $S$ as a function of the current state; thus they are always at least as effective as fixed intelligent adversaries (since any fixed influence set can be implemented as a mobile adversary's policy).
However, it is not clear \emph{a priori} how mobile intelligent adversaries will compare to those of the uniformly random variety.
In this section, we show in Theorem~\ref{thm:MIring} that there exist policies which allow mobile intelligent adversaries to render $\vec{y}$ strictly stochastically stable \emph{much} more easily than the other types, even for relatively low values of $k$.

\begin{theorem}\label{thm:MIring}
Let $G^{\rm r}\in\geer$ be a ring graph that is influenced by a mobile intelligent adversary with capability $k\leq n$. Then
\begin{equation} \label{eq:suscepMI}
\ay\left(G^{\rm r},\MI(k)\right) = \left\{
\begin{array}{cl} 
%0 			& \mbox{if } k=0, 				\vspace{2mm} \\ 
\frac{k}{k+1} 	& \mbox{if } k\in\{1,2\}, 			\vspace{2mm} \\
\frac{n-1}{n}	& \mbox{if } k\in\{3,\ldots,n-1\},		\vspace{2mm} \\
1 			& \mbox{if } k=n.
\end{array}\right.
\end{equation}
\end{theorem}

The proof of Theorem~\ref{thm:MIring} is included in Appendix~\ref{a:intelligent proof}.
Recall that a uniformly random adversary with $k\geq 1$ can stabilize $\vec{y}$ any time $\alpha <1/2;$ an adversary who can intelligently influence a different single agent in $N$ each day can stabilize $\vec{y}$ under these same conditions.
If the mobile intelligent adversary has capability $k\geq 3$, it can stabilize $ \vec{y}$ when $\alpha < (n-1)/n$, i.e., under the same conditions as a fixed intelligent adversary with capability $k=n-1$.

\section{Simulations} \label{sec:sims}

\subsection{Numerically verifying stochastic stability}

Consider fixed intelligent and uniformly random adversaries with~$k=1$ attempting to influence a~$3$-node ring graph, which we denote $G^3$.
Theorems~\ref{thm:ringFI} and~\ref{thm:mr-ring} give the associated susceptibilities as $\ay(G^3,\FI(1))=1/3$ and $\ay(G^3,\UR(1))=1/2$. Thus, for any $\alpha\in(1/3,1/2)$, uniformly random adversaries should be able to render~$\vec{y}$ stochastically stable, while fixed intelligent should not.

To verify this, we numerically compute the stationary distributions
associated with Markov chains for these two adversary types on~$G^3$ with the agents updating actions using log-linear learning.
%\footnote{
%For uniformly random adversaries, we explicitly construct the Markov matrix from probabilities from~\eqref{e:LLL dynamics} (see also the proof of Lemma~\ref{lem:mr} in the Appendix), and raise this matrix to a sufficiently-large power to find the stationary distribution.
%Alternatively, fixed intelligent adversaries induce an exact potential game with potential function
%} %
For each value of noise parameter $\beta$, we record the expected fraction of the agents that are playing $y$ in the respective stationary distribution, and plot these values in Figure~\ref{fig:fracY}.
As expected, when $\alpha<1/3$, both adversary types induce a large fraction of agents to play $y$ when $\beta$ is large, but when $\alpha\in(1/3,1/2)$, only the uniformly-random adversary is able to do so.

\begin{figure}
% this figure generated by https://github.ucsb.edu/philipbrown/adversary-python/blob/sub1/ringPlotScript.py
% with call plotURFIdata(bmax=14,alphas=[1/4,3/8],lss=['-','--'],labels=['UR','FI'],idx=nY/3)
  \centering
    \includegraphics[scale=.25]{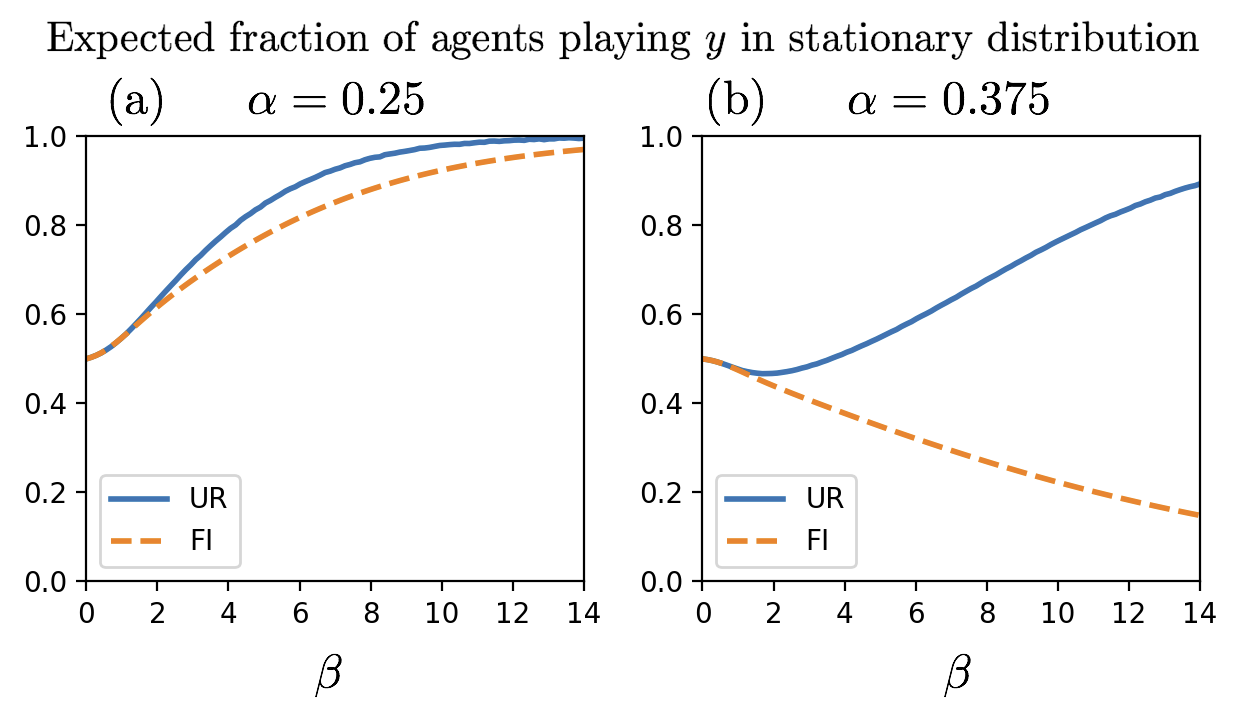}
  \caption{The expected fraction of agents playing $y$ with respect to noise parameter $\beta$, in the stationary distributions of the Markov chains associated with log-linear learning on a $3$-agent ring graph for $\FI(k)$ and $\UR(k)$ with $k=1$.
  In this setting, Theorems~\ref{thm:ringFI} and~\ref{thm:mr-ring} give the susceptibilities of this graph to fixed intelligent and uniformly random adversaries as $1/3$ and $1/2$, respectively.
  On the left, (a) shows that when $\alpha=0.25$ (that is, the graph is susceptible to both types of adversaries), the expected fraction of agents playing $y$ increases with $\beta$, apparently converging to $1$ as $\beta\to\infty$.
  On the right, (b) shows that when $\alpha=0.375$ (that is, the graph is susceptible only to uniformly random adversaries), the two curves diverge from one another; as $\beta$ increases, uniformly-random adversaries convert an increasing fraction of agents to play $y$, while fixed intelligent adversaries convert a decreasing fraction.
   }
  \label{fig:fracY}
\end{figure}

\subsection{Simulations of short-term behavior}

The foregoing work has viewed the problem solely through the lens of stochastic stability, which is primarily concerned with long-run behavior and may not give much indication about the behavior of a system on short timescales.
Accordingly, we perform an empirical hitting-time analysis for all three types of adversary over a wide range of $k$ on a ring graph with $n=10$ and the log-linear learning rule with parameter $\beta=2$.
We compute the average number of iterations it takes each associated Markov process to first reach $\vec{y}$, having been initialized at $\vec{x}$.
For each adversary type, we compute this hitting time as a function of capability $k\in\{1,\ldots,9\}$ and payoff gain $\alpha\in\{0.2,0.3,0.4,0.5,0.6\}$.
See Figure~\ref{fig:allhit} for plots of approximate hitting times, which are computed by averaging over many sample paths for each point.

Figure~\ref{fig:allhit} illustrates that mobile intelligent adversaries have a great advantage over the other types (particularly for low $k$).
The distinction between fixed intelligent and uniformly random types is less clear -- although for low $k$ and large $\alpha$, fixed intelligent adversaries do appear slower than uniformly random.

The plots in Figure~\ref{fig:allhit} suggest that short-term behavior of these Markov chains is not predicted well by stochastic stability.
Nonetheless, several of the general distinctions between adversary types that are visible in the resilience plots of Figure~\ref{fig:comparison} do indeed appear in the hitting times data.
Furthermore, for each $\alpha$, hitting times in Figure~\ref{fig:allhit} appear to converge to nearly the same value for the various adversary types as $k$ approaches $n$.
This should be expected of any performance metric, since for large $k$, the qualitative distinction between the various adversary policies vanishes.

\section{Summary and future work}

This paper represents an initial study into the effects of adversarial influence on graphical coordination games.
A wide array of extensions are possible here.
Naturally, the setting of ring graphs is rather limited, and it is likely that there are many other simple graph structures that will prove analytically tractable (e.g., regular graphs or trees).
A study on the specific impact of graph structure on susceptibility would be interesting; for example, how does the degree distribution of a graph affect its susceptibility?

Looking further ahead, it will be interesting to consider other adversary objectives.
For example, what if rather than stabilizing $\vec{y}$, the adversary simply wanted to minimize the total utility experienced by the agents?
How would changing the objective in this way alter the results?

\bibliographystyle{ieeetr}
\bibliography{../library/library}

\appendix

\subsection{Log-linear learning and its underlying Markov process}\label{a:LLL Markov}

For each model of adversarial behavior, log-linear learning dynamics define a Markov chain, $P_\beta$ over state space $\mathcal{A}$ with transition probabilities parameterized by $\beta>0$~\cite{Blume1993}.
These transition probabilities can readily be computed according to~\eqref{e:LLL dynamics}, taking into account the specifics of the adversarial model in question.
%Note that If $a$ and $a^\prime\in\mathcal{A}$ differ by more than one agent's action, then $P_\beta(a\to a^\prime) = 0$. 
Since $P_\beta$  is aperiodic and irreducible for any $\beta >0$, it has a unique stationary distribution, $\pi_\beta$, with $\pi_\beta P_\beta = \pi_\beta$. 

As $\beta\to\infty$, this converges to a unique limiting distribution
%\begin{equation}
$\pi := \lim_{\beta\to\infty} \pi_\beta.$  
%\end{equation}
If %$\pi(a)>0,$ then $a$ is \emph{stochastically stable}, and if 
$\pi(a) = 1,$ then joint action $a$ is \emph{strictly stochastically stable} \cite{Foster1990}.%
\footnote{Note that this definition of strict stochastic stability is equivalent to the definition in~\eqref{eq:weirdSSS}.}

\begin{figure*}
% data generated with methods in https://github.ucsb.edu/philipbrown/adversary-python/blob/sub1/LLLRing.py 50 to 100 sample paths
% this figure generated by https://github.ucsb.edu/philipbrown/adversary-python/blob/sub1/ringPlotScript.py
% makeGoodPlots(['interesting/FIBlendeach1-9.json','interesting/URBlendeach1-9.json','interesting/MI100each1-9.json'])
  \centering
    \includegraphics[scale=.27]{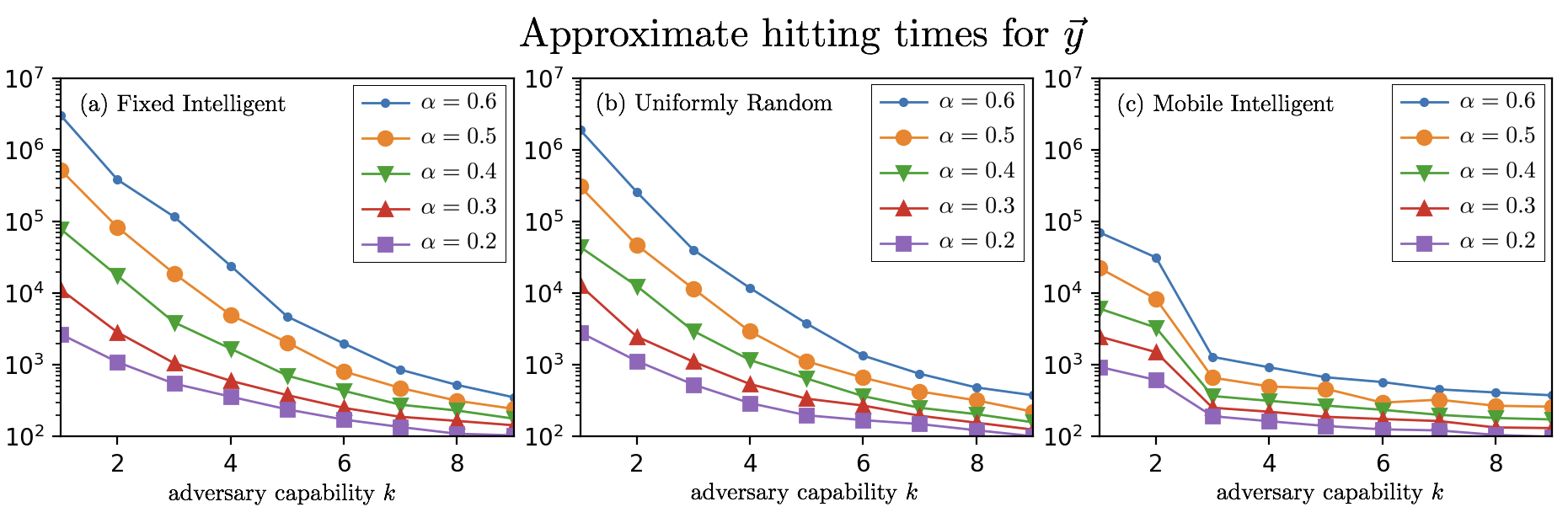}
  \caption{Average number of iterations for the Markov chain associated with log-linear learning ($\beta=2$) on a ring graph ($n=10$) with each type of adversary to reach~$\vec{y}$, having been initialized at~$\vec{x}$, for various values of $k$ and $\alpha$.
  Recall that larger values of $\alpha$ indicate that $x$ has a larger intrinsic advantage over $y$, which makes it more difficult for any adversary to stabilize $\vec{y}$.
  These plots illustrate this clearly: for example, with $k=1$, the uniformly-random adversary takes nearly $2$ orders of magnitude longer to first reach $\vec{y}$ with $\alpha=0.6$ as compared with $\alpha=0.4$.
  Note that the distinction between fixed intelligent and uniformly random adversaries is difficult to see, but that the mobile intelligent adversary outperforms the others for low $k$ and large $\alpha$ by at least an order of magnitude.
  Also, for $k=1$ and $\alpha\geq0.4$, these plots do clearly indicate that hitting times for fixed intelligent adversaries are at least twice as large as for the uniformly-random variety.
 For the mobile intelligent adversaries, as might be expected from Figure~\ref{fig:comparison}, there is a dramatic improvement in the adversary's performance going from $k=2$ to $k=3$, while the gains for larger $k$ are more modest.
  }
  \label{fig:allhit}
\end{figure*}

As $\beta\to \infty$, transition probabilities $P_\beta(a\to a^\prime)$ of  log-linear learning converge to the transition probabilities, $P(a\to a^\prime)$, of a best response process.  
Distribution $\pi$ is one of possibly multiple stationary distributions of a best response process over game $G$. %Thus, log-linear learning selects a stationary distribution of the best response process. 

\subsection{Stability in the presence of a fixed intelligent adversary}\label{a:fixed line graph proofs}

When a fixed intelligent adversary influences set $S$, the corresponding influenced graphical coordination game is a potential game \cite{Monderer1996} with potential function 
\begin{equation}
\Phi^S(a_i,a_{-i}) = {1\over 2}\sum_{i\in N}\left(U_i(a_i,a_{-i}) + 2\cdot\mathds{1}_{i\in S, a_i = y}\right).
\end{equation}
It is well known that $a\in\mathcal{A}$ is strictly stochastically stable if and only if $\Phi^S(a) >\Phi^S(a^\prime)$ for all $a^\prime\in\mathcal{A},$ $a^\prime\neq a$~\cite{Blume1993}. %This fact may be used to prove Theorems~\ref{t:A stable all} - \ref{thm:ringFI}.

%The proof of Theorem~\ref{thm:ringFI} follows using \eqref{e:stationary distn}.

\vspace{2mm}
\noindent\emph{Proof of Theorem~\ref{thm:genFI}:}
This proof uses techniques inspired by \cite[Proposition 2]{Young2011}. 
Let ${G} = (N,E)$ and suppose the fixed intelligent adversary is influencing set $S\subset N$.
Let
$(\vec{y}_T,\vec{x}_{N\setminus T})$ denote the joint action $a$ with $T = \{i\st a_i = y\}.$

If $\alpha$ satisfies~\eqref{eq:thm1} for all $T\subseteq N$ with $T\neq\emptyset$, we have
\begin{align}
\Phi^S(\vec{x})  &= (1+\alpha)d(N,N) \label{eq:phix} \\
&> (1+\alpha)d(N\setminus T,N\setminus T) + d(T,T) + |T\cap S| &\nonumber\\
&= \Phi^S(\vec{y}_T,\vec{x}_{N\setminus T}),\nonumber 
\end{align}
where the inequality follows from~\eqref{eq:thm1}, the definition of $d(S,T)$, and basic set theory.
This implies that joint action $\vec{x}$ is strictly stochastically stable.
The opposite implication follows easily, and~\eqref{eq:thm1} is proved.

The proof of~\eqref{eq:thm2} follows an identical argument, replacing $\vec{x}$ in~\eqref{eq:phix} with $\vec{y}$ and noting that $\Phi^S(\vec{y})=d(N,N)+k$.
\hfill\qed

\vspace{2mm}
\noindent\emph{Proof of Theorem~\ref{thm:ringFI}:}
Suppose $\alpha < {k\mathop{/}n}.$
Then 
$$\Phi^S(\vec{y}) = n+k > n + \alpha n = \Phi^S(\vec{x})$$
for any $S\subseteq N$ with $|S| = k.$
Then, to show that $\vec{y}$ is stochastically stable for influenced set $S$ satisfying~\eqref{eq:FIpolicy},
it remains to show that 
$\Phi^S(\vec{y}) > \Phi^S(\vec{y}_T,\vec{x}_{N\setminus T})$
for any $T\subset N$ with $T\neq \emptyset$ and $T\neq N.$
Suppose the graph restricted to set $T$ has $p$ components, where $p\geq 1.$
Label these components as $T_1,T_2,\ldots,T_p$ and define $t:=|T|$ and $t_i:=|T_i|$. %Let $\ell$ represent the number of components in the graph restricted to $N\setminus T.$ %Since this is a ring graph, we have $\ell\in \{p-1,p,p+1\}$, and since $T\neq N$, $\ell \geq 1.$ 
For any $T\subset N$ with $T\neq N,T\neq\emptyset,$ and $0<t<n$,
\begin{align}
\Phi^S(\vec{y}_{T},\vec{x}_{N\setminus T}) &=(1+\alpha)(n-t-p) + t-p + \sum_{j=1}^p |S\cap T_j|  \nonumber\\ 
&< n+k \nonumber\\ %\label{e:some label}\\ %\displaybreak[3]\\
&=\Phi^S(\vec{y}),\nonumber
\end{align}
where the inequality follows from~$\alpha<k/n$ and~\eqref{eq:FIpolicy} combined with the fact that $\lceil z \rceil \leq z+1$ for any $z\in \mathbb{R}$.
That is, $\ay(G,\FI(k))\geq k/n$.
The matching upper bound is given by Corollary~\ref{cor:fid-general}, since in a ring graph $|E|=n$.
\hfill\qed

\subsection{Resistance trees for stochastic stability analysis}\label{a:resistance trees}

When graphical coordination game $G$ is influenced by a uniformly-random adversary, it is no longer a potential game; resistance tree tools defined in this section enable us to determine stochastically stable states for uniformly-random and mobile intelligent adversaries.

The Markov process $P_\beta$ defined by log-linear learning dynamics over a normal form game is a \emph{regular perturbation} of a best response process.
In particular, log-linear learning is a regular perturbation of the best response process defined in Appendix~\ref{a:LLL Markov}, where the size of the perturbation is parameterized by $\eps = e^{-\beta}.$
The following definitions and analysis techniques are taken from \cite{Young1993}. 

\begin{defn}[Regular Perturbed Process~\cite{Young1993}]
A Markov process with transition matrix $M_\eps$ defined over state space $\Omega$ and parameterized by perturbation $\eps\in (0,a]$ for some $a>0$ is a \emph{regular perturbation} of the process $M_0$ if it satisfies:
\begin{enumerate}[leftmargin=1.5em]
\item $M_\eps$ is aperiodic and irreducible for all $\eps\in (0,a]$.
\item $\lim_{\eps\to 0^+}M_\eps (\xi,\xi') \to M(\xi,\xi')$ for all $\xi,\xi'\in\Omega.$
\item If $M_\eps(\xi,\xi') >0$ for some $\eps \in (0,a]$ then there exists $r(\xi,\xi')$ such that 
\begin{equation}\label{e:resistance}
0<\lim_{\eps\to 0^+} {M_\eps(\xi,\xi')\over \eps^{r(\xi,\xi')}} <\infty,
\end{equation}
where $r(\xi,\xi')$ is referred to as the \emph{resistance} of transition $\xi\to \xi'$.
\end{enumerate}
\end{defn}

Let Markov process $M_\eps$ be a regular perturbation of process $M_0$ over state space $\Omega$, where perturbations are parameterized by $\eps\in (0,a]$ for some $a>0.$
Define graph $\Gamma = (\Omega,\mathcal{E})$ to be the directed graph with $(\xi,\xi')\in \mathcal{E} \iff M_\eps(\xi,\xi')>0$ for some $\eps\in(0,a]$.
Edge $(\xi,\xi')\in \mathcal{E}$ is weighted by the resistance $r(\xi,\xi')$ defined in \eqref{e:resistance}.
The resistance of path $p = (e_1,e_2,\ldots,e_k)$ is the sum of the resistances of the associated state transitions:
\begin{equation}
r(p):= \sum_{\ell = 1}^k r(e_\ell).
\end{equation}

Now let $\Omega_1,\Omega_2,\ldots,\Omega_m$ denote the $m\geq 1$ recurrent classes of process $M_0.$ In graph $G$, these classes satisfy:
\begin{enumerate}[leftmargin=1.5em]
\item For all $\xi\in \Omega$, there is a zero resistance path in $\Gamma$ from $\xi$ to $\Omega_i$ for some $i\in \{1,2,\ldots,m\}.$
\item For all $i\in \{1,2,\ldots,m\}$ and all $\xi,\xi'\in \Omega_i$ there exists a zero resistance path in $\Gamma$ from $\xi$ to $\xi'$ and from $\xi'$ to $\xi$.
\item For all $\xi,\xi'$ with $\xi\in \Omega_i$ for some $i\in \{1,2,\ldots,m\},$ and $\xi'\notin \Omega_i$, $r(\xi,\xi') >0.$ 
\end{enumerate}

Define a second directed graph $\mathcal{G} = (\{\Omega_i\},\mathscr{E})$ over the $m$ recurrent classes in $\Omega.$
This is a complete graph; i.e., $(i,j)\in \mathscr{E}$ for all $i,j\in \{1,2,\ldots,m\},$ $i\neq j.$
Edge $(i,j)$ is weighted by $R(i,j)$, the total resistance of the lowest resistance path in $\Gamma$ starting in $\Omega_i$ and ending in $\Omega_j$:
\begin{equation}
R(i,j):=\min_{i\in \Omega_i,j\in\Omega_j}\min_{p\in \mathcal{P}(i\to j)} r(p),
\end{equation}
where $\mathcal{P}(i\to j)$ denotes the set of all simple paths in $\Gamma$ beginning at $i$ and ending at $j$. 

Let $\mathcal{T}_i$ be the set of all spanning trees of $\mathcal{G}$ rooted at $i$.
Denote the \emph{resistance} of tree $T\in\mathcal{T}_i$ by
%\begin{equation}
$R(T) := \sum_{e\in T} R(e),$
%\end{equation}
and define 
\begin{equation}
\gamma_i := \min_{T\in \mathcal{T}_i} R(T)
\end{equation}
to be the \emph{stochastic potential} of $\Omega_i$.
We use the following theorem due to \cite{Young1993} in our analysis:
\begin{theorem}[From \cite{Young1993}]\label{t:resistance trees theorem}
State $\xi\in \Omega$ is stochastically stable if and only if $\xi\in \Omega_i$ where 
\begin{equation}
\gamma_i = \min_{j\in \{1,2,\ldots,m\}}\gamma_j,
\end{equation}
i.e., $x$ belongs to a recurrent class with minimal stochastic potential.
Furthermore, $\xi$ is strictly stochastically stable if and only if $\Omega_i = \{\xi\}$ and 
$\gamma_i<\gamma_j,\quad\forall j\neq i.$
\end{theorem}

\vspace{-3mm}
\subsection{Stability in the presence of a uniformly random adversary}\label{a:mobile random}
The following lemma applies to any graphical coordination game in the presence of a uniformly random adversary with capability $k\leq n-1$.
It states that a uniformly random adversary decreases the resistance of transitions when an agent in $N$ changes its action from $x$ to $y$, but does not change the resistance of transitions in the opposite direction.
Intuitively, this means that viewed through the lends of transition resistances, a uniformly-random adversary spreads $y$ throughout the network optimally, but does nothing to slow the spread of $x$.
In a sense, uniformly-random adversaries are all offense and no defense.

\begin{lemma}\label{lem:mr}
Suppose agents in $N$ update their actions according to log-linear learning in the presence of a uniformly random adversary with capability $k$, where $1\leq k\leq n-1.$
Then the resistance of a transition where agent $i\in N$ changes its action from $x$ to $y$ is:
\begin{align}
r((x,a_{-i})\to (y,a_{-i}) )%\nonumber\\
=  \max\left\{U_i(x,a_{-i}) - U_i(y,a_{-i}) -1,0\right\}
\end{align}
and the resistance of a transition where agent $i\in N$ changes its action from $y$ to $x$ is:
\begin{align}
r((y,a_{-i})\to (x,a_{-i}) )%\nonumber\\
=  \max\left\{U_i(y,a_{-i}) - U_i(x,a_{-i}), 0\right\}.\label{e:BtoA1}
\end{align}
Here $U_i:\mathcal{A}\to\R$, defined in \eqref{e:original utility}, is the utility function for agent $i$ in the uninfluenced game, $G$.  
\end{lemma}

\noindent\emph{Proof:}
In the presence of a uniformly random adversary, %the probability that agent $i\in N$ changes its action from $x$ to $y$ is:
\begin{align*}
P_\beta \left((x,a_{-i})\to (y,a_{-i})\right) %\\
= {1\over n} \left({k\over n}\cdot{\exp( \beta(U_i(y,a_{-i})+1))\over \exp( \beta (U_i(y,a_{-i})+1)) + \exp( \beta U_i(x,a_{-i}))}\right.\\
+ \left.{n-k\over n}\cdot{\exp( \beta U_i(y,a_{-i}))\over \exp( \beta U_i(y,a_{-i})) + \exp( \beta U_i(x,a_{-i}))}\right)
\end{align*}
Define $P_\eps\left((x,a_{-i})\to (y,a_{-i})\right)$ by substituting $\eps = e^{-\beta}$ into the above equation.
Algebraic manipulations yield
\begin{equation*}
0<\lim_{\eps\to 0^+}{P_{\eps} \left((x,a_{-i})\to (y,a_{-i})\right)\over \eps^{  U_i(x,a_{-i}) - U_i(y,a_{-i}) -1 }}<\infty,
\end{equation*}
implying 
\begin{align*}
r((x,a_{-i})\to (y,a_{-i}) ) %\nonumber\\
=  \max\left\{U_i(x,a_{-i}) - U_i(y,a_{-i}) -1,0\right\}.
\end{align*}
Equation \eqref{e:BtoA1} may be similarly verified.
\hfill\qed

\vspace{2mm}
\noindent\emph{Proof of Theorem~\ref{thm:ur-gen}: }\label{a:proof mr ss}
By Lemma~\ref{lem:mr}, for any graphical coordination game with graph $G$, the resistance graph associated with log linear learning is the same for all $k\leq n-1$. 
Thus it follows that the set of stochastically stable states is independent of $k$, and thus the susceptibility~$\ay(G,\UR(k))$ is as well.
\hfill\qed

\vspace{2mm}
\noindent\emph{Proof of Theorem~\ref{thm:mr-ring}: }\label{a:proof mr ss}
For any $\alpha\in(0,1)$, $\vec{x}$ and $\vec{y}$ are the only recurrent classes of the unperturbed process $P$.
We call a state transition an \emph{incursion} if it is either $\vec{x}\rightarrow a$ or $\vec{y}\rightarrow a$.
Let $a,a'\in\aaa$ differ in the action of a single agent, and suppose that $a\notin\{\vec{x},\vec{y}\}$; that is, $a\rightarrow a'$ is \emph{not} an incursion.
Lemma~\ref{lem:mr} gives that $r(a\rightarrow a')=0$.
On the other hand, all incursions have positive resistance: $r(\vec{x}\rightarrow a)=2$, and $r(\vec{y}\rightarrow a)=1+2\alpha$.
Thus, there is a sequence of $0$-resistance transitions from any $a\in\aaa$ to both $\vec{x}$ and $\vec{y}$, but every transition out of one of these states has positive resistance -- implying that they are the only recurrent classes of $P$.

Since all transitions other than incursions have $0$ resistance, the problem of finding the stochastic potential of $\vec{x}$ boils down to measuring the resistance of \emph{leaving} $\vec{y}$, and vice versa.
That is, $R(\vec{x}\rightarrow\vec{y})=1+2\alpha$ and $R(\vec{y}\rightarrow\vec{s})=2$.
Thus,
\begin{align}
\alpha<1/2\ &\implies \ R(\vec{x}\rightarrow\vec{y})<R(\vec{y}\rightarrow\vec{s}) \\
&\ \ \mbox{and} \nonumber \\
\alpha>1/2\ & \implies \ R(\vec{x}\rightarrow\vec{y})>R(\vec{y}\rightarrow\vec{s}),
\end{align}
yielding the proof.
\hfill\qed

\vspace{-2mm}
\subsection{Stability in the presence of a mobile intelligent adversary}\label{a:intelligent proof}

Similar to Lemma~\ref{lem:mr}, in Lemma~\ref{lem:mi} we provide a characterization of the resistances of state transitions in the presence of a mobile intelligent adversary.
Naturally, these resistances are a function of the adversary's policy, and thus unlike the resistances due to a uniformly random adversary, they do depend implicitly on $k$.

\begin{lemma}\label{lem:mi}
Suppose agents in $N$ update their actions according to log-linear learning in the presence of a mobile intelligent adversary with capability $k\in\{1,\ldots,n-1\}$ and policy $S:\aaa\rightarrow \sss_k$.
Then the resistance of a transition where agent $i\in N$ changes its action from $x$ to $y$ in the presence of policy $S$ is:
\begin{align}
r^S((x,a_{-i})\to (y,a_{-i}) )%\nonumber\\
&=  \max\left\{\tilde{U}_i(x,a_{-i},S(a)) - \tilde{U}_i(y,a_{-i},S(a)),0\right\}
\end{align}
and the resistance of a transition where agent $i\in N$ changes its action from $y$ to $x$ in the presence of policy $S$ is:
\begin{align}
r^S((y,a_{-i})\to (x,a_{-i}) )%\nonumber\\
=  \max\left\{\tilde{U}_i(y,a_{-i},S(a)) - \tilde{U}_i(x,a_{-i}S(a)), 0\right\}.\label{e:BtoA1}
\end{align}
Here $\tilde{U}_i:\mathcal{A}\to\R$, defined in \eqref{e:new utility}, is the utility function for agent $i$ in the influenced game $\tilde{G}$.  
\end{lemma}

\begin{proof}
This proof proceeds in exactly the same fashion as that of Lemma~\ref{lem:mr}, except that since a mobile intelligent adversary's policy is deterministic, if $i\in S$, the adversary modifies $U_i$ with probability $1$.
Thus, the asymmetry between $x\to y$ and $y\to x$ transitions that appeared in Lemma~\ref{lem:mr} vanishes.
\end{proof}

For ring graphs, Table~\ref{tab:resist} enumerates several of the key transition resistances experienced under the influence of mobile intelligent adversaries.

Next, for the special case of $k=2$, the following lemma fully characterizes the recurrent classes of the unperturbed best response process $P^{\bar{s}}$ associated with an arbitrary mobile intelligent adversary's policy $\bar{S}$.

\begin{lemma} \label{lem:k2}
For any ring graph $G\in\geer$, let $k=2$, let $\sbar$ be any adversary policy, and $P^{\bar{S}}$ be the unperturbed best response process on $G$ associated with $\bar{S}$.
Every recurrent class of $P^{\bar{S}}$ is a singleton.
Furthermore, an action profile $\abar$ (other than $\vec{x}$ and $\vec{y}$) is a recurrent class of $P^{\bar{S}}$ if and only if the following two conditions are both satisfied:
\begin{enumerate}
\item $\abar$ has a single contiguous chain of agents $[i,j]$ playing $x$: $\abar_i=\abar_{i+1}=\cdots=\abar_j=x$ such that $j-i\in[1,n-3]$.
\item $\sbar(\abar) = \{i-1,j+1\}$.
\end{enumerate}
\end{lemma}

\begin{proof}
At a state $\abar$ and policy satisfying  1) and 2), either agent $i$ or $j$ can switch to $y$ with a resistance of $\alpha$, or agent $i-1$ or $j+1$ can switch to $x$ with a resistance of $1-\alpha$.
All other transitions have higher resistance than these; this demonstrates that $\abar$ is a singleton recurrent class.
We will show the contrapositive to complete the proof of 1) and 2).
If $\abar$ has more than one contiguous chain of agents playing $x$, then either there is an agent playing $y$ with two neighbors playing $x$, or there is an uninfluenced agent playing $y$ with a neighbor playing $x$.
Thus, that agent deviating to $x$ is a $0$-resistance transition to another state (see Table~\ref{tab:resist}) so $\abar$ cannot be a singleton recurrent class.
If $\abar$ has a unique contiguous chain of $x$ of length $1$ or $n-1$, then there is a $0$-resistance transition to $\vec{y}$ or $\vec{x}$, respectively.
If 1) is satisfied but 2) is not, then there is an uninfluenced agent playing $y$ with a neighbor playing $x$; again a $0$-resistance transition to another state.

To see that every recurrent class is a singleton, simply note that any state other than $\vec{x}$, $\vec{y}$, or one satisfying 1) can be transformed into either~$\vec{x}$ or~$\vec{y}$ by a finite sequence of $0$-resistance transitions.
\end{proof}

Naturally, there are a vast number of policies which a mobile intelligent adversary can employ; we propose a family of such policies for ring graphs; we say that any policy in this family is a \emph{balanced policy.}

\begin{table}
\setlength{\extrarowheight}{1pt}
\caption{ Summary of transition resistances derived from Lemma~\ref{lem:mi}}
\begin{center}
	\begin{tabular}{c|c||c|c}\label{tab:resist}
	\# nbrs with $a_j=a_i$ & $i\in S$ & $r^S(x\to y)$  & $r^S(y\to x)$  \\ \hline
	\multirow{2}{*}{0}		& no 		& $2+2\alpha$ 		& $2$ \\
						& yes	& $1+ 2\alpha$		& $3$ \\ \hline
	\multirow{2}{*}{1}		& no 		& $\alpha$ 		& $0$ \\
						& yes	& $0$			& $1-\alpha$ \\ \hline
	\multirow{2}{*}{2}		& no 		& $0$	 		& $0$ \\
						& yes	& $0$			& $0$ \\ \hline
	\end{tabular}
	\vspace{-2mm}
\end{center}
\end{table}

\begin{defn}[Balanced Policy]\label{def:optpolicy}
Let $G\in\geer$ be a ring graph with $n$ nodes.
For any state $a\in\aaa,$ let $[i,j]$ be the longest chain of agents playing $x$ (break ties lexicographically).
For a type $\MI(k)$ agent, policy $S:\aaa\rightarrow\sss_k$ is \emph{balanced} if it satisfies the following conditions.%
\footnote{
Note that Definition~\ref{def:optpolicy} does not always specify the location of every adversary.
Thus, there is a large family of policies satisfying Definition~\ref{def:optpolicy}.
}
\begin{enumerate}
\item If $j-i>1$, let $i\in S(a)$. If additionally $k\geq3$, let $\{i-1,i,j+1\}\subseteq S(a)$.
\item If $k\geq2$ and $j=i$, let $\{i-1,i+1\}\subseteq S(a)$.
\end{enumerate}
\end{defn}

\noindent See Figure~\ref{fig:balanced} for a graphical depiction of the key elements of this policy.
The key idea is that there should always be one adversary ``attacking'' (red circles in Figure~\ref{fig:balanced}) an $x$ who has a $y$ neighbor, and if there are enough adversaries, the longest contiguous chain of $x$'s should always be surrounded by a pair of ``defensive'' adversaries (green circles in Figure~\ref{fig:balanced}).
That is, this policy specifies the placement of no more than $3$ adversaries; the placement of any remaining ``indeterminate'' adversaries (dark gray circles in Figure~\ref{fig:balanced}) is of no importance to the results of Theorem~\ref{thm:MIring}.

\vspace{2mm}
\noindent We can now proceed with the proof of the theorem.

\noindent\emph{Proof of Theorem~\ref{thm:MIring}:}
We will show that under a balanced policy, if $\alpha$ is less than each of the susceptibilities shown in~\eqref{eq:suscepMI} that $\vec{y}$ is strictly stochastically stable; subsequently, we will show that no policy can outperform a balanced policy.
In the following, we write $R^S(a,a')$ to denote the resistance of a transition from $a$ to $a'$ in the presence of adversary policy $S:\aaa\to\sss_k$.

\subsubsection*{Susceptibility to a balanced policy}
Let $G = (N,E)$ be a ring graph influenced by a mobile intelligent adversary with capability $k$ using a balanced policy $\sbal$ satisfying Definition~\ref{def:optpolicy}.
As in the proof of Theorem~\ref{thm:mr-ring}, only $\vec{x}$ and $ \vec{y}$ are recurrent in the unperturbed process $P^{\sbal}$, because at every state $a\in\aaa$, there is a sequence of $0$-resistance transitions leading to either $\vec{x}$ or $ \vec{y}$, as can be verified by the resistances listed in Table~\ref{tab:resist}.

First consider $R^\sbal(\vec{x}\rightarrow\vec{y})$.
Lemma~\ref{lem:mi} gives that the incursion of $y$ into state $\vec{x}$ (at an agent being influenced by an adversary) has resistance $r^\sbal(\vec{x}\rightarrow a)=1+2\alpha$.
Then, for any state $a\notin\{\vec{x},\vec{y}\}$, there is always at least one agent playing $x$ who has either two neighbors playing $y$, or a neighbor playing $y$ and is connected to an adversary.
Thus, there is always a sequence of transitions from $a$ to $\vec{y}$ with a total resistance of $0$, so $R^\sbal(\vec{x}\rightarrow\vec{y})=1+2\alpha$.

%\addtolength{\textheight}{-4.75in}   	% This command serves to balance the column lengths

Next consider $R^\sbal(\vec{y}\rightarrow\vec{x})$.
Whenever $k<n$, there is always at least one agent that is not being influenced by the adversary; thus Lemma~\ref{lem:mi} gives that an incursion of $x$ into $\vec{y}$ has a resistance of $2$.
If $k=1$, a balanced policy does not allow the adversary to ``play defense''; so there is a sequence of subsequent $y\rightarrow x$ transitions that each have $0$ resistance.
Thus when $k=1$ with a balanced policy, the situation is identical to that of the uniformly-random adversary, and $\vec{y}$ is strictly stochastically stable whenever $\alpha<1/2$.

If $k\geq2$, the first $y\rightarrow x$ transition \emph{after} the incursion now has a positive resistance of at least $1-\alpha$.
If $k=2$, the policy does not allow the adversary to protect against further spread of $x$, so we have that $R^\sbal(\vec{y}\rightarrow\vec{x})=3-\alpha$.
That is, whenever $\alpha<2/3$, we have that $R^\sbal(\vec{x}\rightarrow\vec{y})<R^\sbal(\vec{y}\rightarrow\vec{x})$ so $\vec{y}$ is strictly stochastically stable.

\begin{figure}
% this figure generated by https://github.ucsb.edu/philipbrown/adversary-python/blob/master/plotscript.py
  \centering
    \includegraphics[scale=.15]{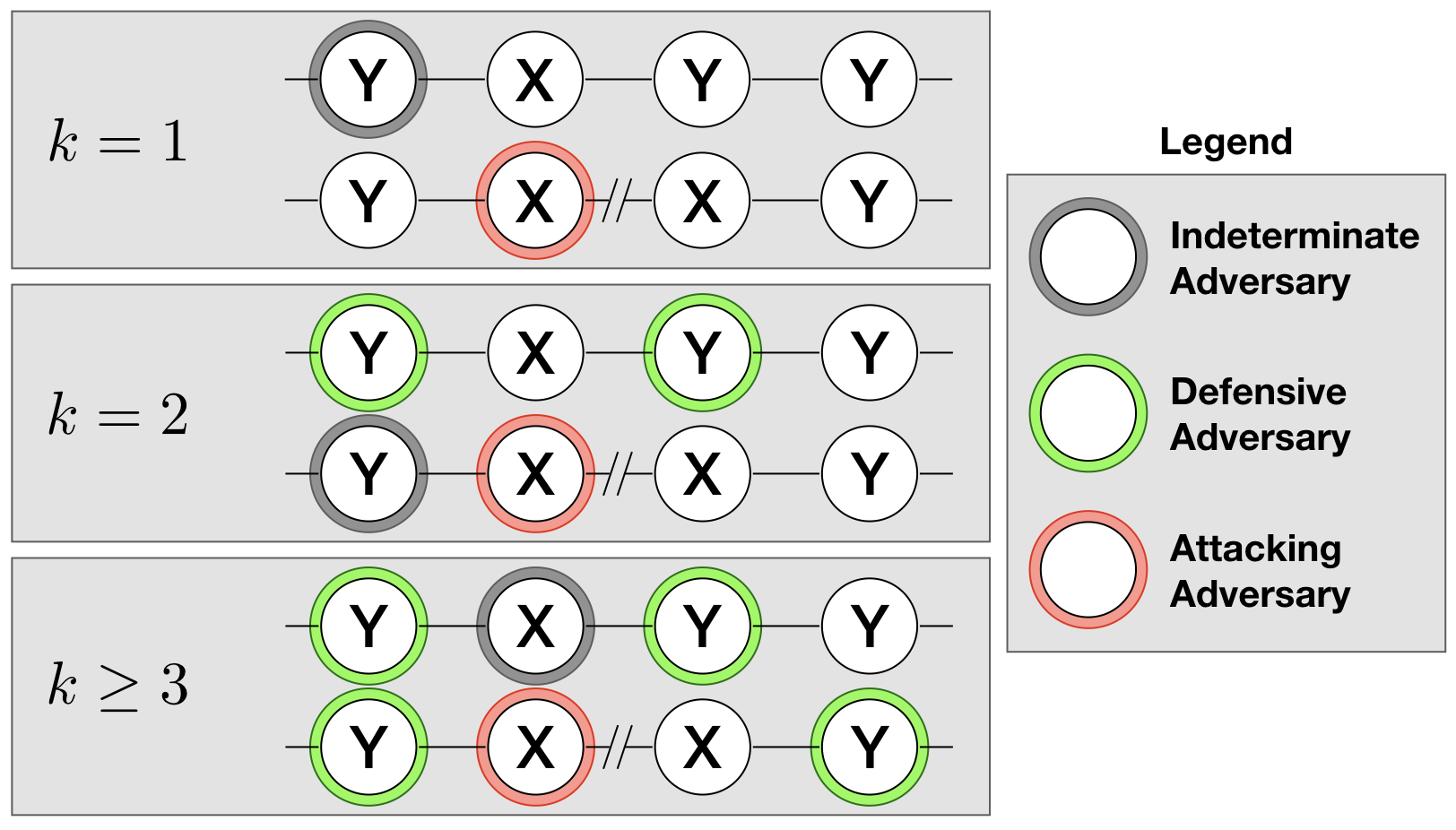}
  \caption{Graphical depiction of a balanced mobile intelligent adversary policy as defined by Definition~\ref{def:optpolicy}.
  There are essentially three strategies of adversarial influence: Indeterminate (gray), Defensive (green), or Attacking (red).
  For each $k$ depicted above, the upper chain of agents depicts the special case that a single agent is playing $x$; the lower depicts the longest contiguous chain in the graph of agents playing $x$.
  ``Attacking'' means influencing an agent $i$ with $a_i=x$ and a neighbor playing $y$.
  ``Defensive'' means influencing the two agents playing $y$ who are at the ends of the longest chain of $x$.
  ``Indeterminate'' means that the adversary's choice of which agent to influence has no impact on the stochastically-stable states of the system.
  Note that for all $k\geq1$, a balanced policy requires one of the $k$ adversaries to be attacking whenever the longest chain of $x$ is longer than $1$.
  For $k=2$, in the special case that the longest chain of $x$ is length $1$, both adversaries are defensive.
  For $k\geq 3$, there are always enough adversaries that two can be defensive and one attacking.}
  \label{fig:balanced}
\end{figure}

If $k\geq3$, the policy always defends against further spread of the longest contiguous chain of agents playing $x$.
Consider the $\vec{y}\rightarrow\vec{x}$ path given by $\{1,2,\ldots,n\}$; under a balanced policy, this path has resistance $2+(n-2)(1-\alpha).$
Note that any alternative $\vec{y}\rightarrow\vec{x}$ path has resistance no less than $4$, as it would require at least two $y\rightarrow x$ transitions of agents who have no neighbors playing $x$.
For any $\alpha<3/2$, $R^\sbal(\vec{x}\rightarrow\vec{y})<4$; thus, the only relevant situation is when $R^\sbal(\vec{y}\rightarrow\vec{x})=2+(n-2)(1-\alpha)$.
When this is the case, whenever $\alpha<(n-1)/n$, it follows that $R^\sbal(\vec{x}\rightarrow\vec{y})<R^\sbal(\vec{y}\rightarrow\vec{x})$ so $\vec{y}$ is strictly stochastically stable.

\subsubsection*{A balanced policy is optimal}
Whenever $k<n$, for any ring graph $G$, it is clear that $\ay(G,\MI(k))\leq\ay(G,\MI(n-1))$ because adding additional adversaries can only decrease $R(\vec{x}\rightarrow\vec{y})$ and/or increase $R(\vec{y}\rightarrow\vec{x})$.
Furthermore, it always holds that $\ay(G,\MI(n-1))\leq\ay(G,\FI(n-1))=(n-1)/n$.
This is because mobile intelligent adversaries are strictly more capable than fixed intelligent.
Thus, the susceptibility in~\eqref{eq:suscepMI} is tight for the case of $k\geq3$.

When $k=1$, for any adversary policy $S_1$ of type $\MI(1)$, the only recurrent classes are $\vec{x}$ and $\vec{y}$.
This is because the adversary does not have enough capability to ``play defense'': at every state $a$ such that some agent is playing $x$, then regardless of the adversary's policy, there is an agent playing $y$ that can switch to $x$ with $0$ resistance (see Table~\ref{tab:resist}).
That is, there is a $0$-resistance transition to a state with strictly more agents playing $x$, showing that $R^{S_1}(\vec{y}\to\vec{x})= R^\sbal(\vec{y}\rightarrow\vec{x})$.
It has already been shown that $R^{S_1}(\vec{x}\rightarrow\vec{y})\geq R^\sbal(\vec{x}\rightarrow\vec{y})$, and thus it is proved that a balanced policy with $k=1$ is optimal.

When $k=2$, the situation is more challenging, because there are too few adversaries for us to appeal to Theorem~\ref{thm:ringFI}'s upper bound, but enough adversaries that the unperturbed process may have a multiplicity of recurrent classes.
Note that to show that $\ay(G,\MI(2))\leq2/3$, it suffices to show that when $\alpha=2/3$, $\vec{y}$ can never be strictly stochastically stable for any adversary policy.

Let $S_2$ be any adversary policy with $k=2$, and suppose there are $m$ states satisfying both conditions of Lemma~\ref{lem:k2} (we call these the ``mixed'' recurrent classes; in each, some agents are playing $x$ and some $y$).

When $\alpha=2/3$, the minimum-resistance tree rooted at $\vec{x}$ has total resistance no more than $2+(m+1)(1-\alpha)$: the resistance of leaving $\vec{y}$ is $2$, the resistance of leaving each of the $m$ mixed recurrent classes is $1-\alpha$, and there can be up to an additional resistance-$1-\alpha$ transition from $\vec{y}$ to a state with $2$ agents playing $x$.
Thus, the stochastic potential of $\vec{x}$ as a function of $\alpha$ is
\begin{equation}
\gamma_{\vec{x}}(\alpha)= 2+(m+1)(1-\alpha). \label{eq:gx}
\end{equation}

Let $\abar$ denote a mixed recurrent class, and with abuse of notation, also the state associated with that class.
When $\alpha=2/3$, for any other state $a'$ accessible from $\abar$, Table~\ref{tab:resist} gives us that $r^{S_2}(\abar\to a')\geq 1-\alpha$.
Thus, for any other recurrent class $a^\dagger$, we have 
\begin{equation}
R^{S_2}\left(\abar,a^\dagger\right)\geq 1-\alpha. \label{eq:Rbound}
\end{equation}

Now, let $T$ be the minimum-resistance tree rooted at $\vec{y}$.
For any recurrent class $a^\dagger$, $R^{S_2}(\vec{x},a^\dagger)\geq1+2\alpha$, so it follows from~\eqref{eq:Rbound} that the stochastic potential of $\vec{y}$ as a function of $\alpha$ is lower bounded by
\begin{align}\label{eq:gy}
\gamma_{\vec{y}}(\alpha) = R^{S_2}(T) 		&\geq 1+2\alpha + m (1-\alpha). 
\end{align}

It can be readily seen that for all $\alpha>2/3$, $\gamma_{\vec{x}}(\alpha) < \gamma_{\vec{y}}(\alpha)$, indicating that $\vec{x}$ is strictly stochastically stable, and showing that $\ay(G^{\rm r},\MI(2))\leq2/3$.
\hfill\qed

\end{document}